\documentclass[10pt,conference]{IEEEtran}
\usepackage{ifpdf}
\usepackage{manfnt}
\usepackage[mathscr]{eucal}
\usepackage{graphicx}

\usepackage{algorithmicx}
\usepackage[ruled,vlined,commentsnumbered]{algorithm2e}
\usepackage{amssymb}
\usepackage{theorem}
\usepackage{cite}
\usepackage{color}
\usepackage{colortbl}
\definecolor{colorref}{rgb}{0.4648,0,0} 
\definecolor{colorcite}{rgb}{0,0.2902,0.1765}

\usepackage{bm}

\usepackage{amsmath}
\newtheorem{proposition}{Proposition}

\newtheorem{lemma}{Lemma}

\newcommand{\Rmnum}[1]{\uppercase\expandafter{\romannumeral #1}}

\renewcommand{\baselinestretch}{1.33}

\newcommand{\signal}[1]{{\boldsymbol{#1}}}

\newcommand{\real}{{\mathbb R}}

\newtheorem{definition}{Definition}

\newtheorem{fact}{Fact}

\newtheorem{problem}{Problem}

\newcommand{\Natural}{{\mathbb N}}
\newcommand{\refeq}[1]{(\ref{#1})}

\newif\iftodo   
\todotrue
\newif\iftodoshort  
\todoshortfalse

\title{Performance Limits of Solutions to \\ Network Utility Maximization Problems}
\author{R.~L.~G.~Cavalcante and S.~Sta\'nczak\\ Fraunhofer Heinrich Hertz Institute and Technical University of Berlin \\ 
	email: \{renato.cavalcante,slawomir.stanczak\}@hhi.fraunhofer.de}

\newenvironment{proof}{{\it Proof:}}{\hfill$\square$\\}

\begin{document}
\maketitle

\begin{abstract}
We study performance limits of solutions to utility maximization problems (e.g., max-min problems) in wireless networks as a function of the power budget $\bar{p}$ available to transmitters. Special focus is devoted to the utility and the transmit energy efficiency (i.e., utility over transmit power) of the solution. Briefly, we show tight bounds for the general class of network utility optimization problems that can be solved by computing conditional eigenvalues of standard interference mappings.  The proposed bounds, which are based on the concept of asymptotic functions, are simple to compute, provide us with  good estimates of the performance of networks for any value of $\bar{p}$ in many real-world applications, and enable us to determine points in which networks move from a noise limited regime to an interference limited regime.  Furthermore, they also show that the utility and the transmit energy efficiency scales as $\Theta(1)$ and $\Theta(1/\bar{p})$, respectively, as $\bar{p}\to\infty$.
\end{abstract}

\section{Introduction}
Recent studies \cite{nuzman07,renato2016maxmin,sun2014,renatomaxmin} have shown a strong connection between the solution to many utility maximization problems in wireless networks and conditional eigenvalues of nonlinear mappings \cite{krause01} (typically, standard interference mappings \cite{yates95}).  In general, the utility of the solution increases as we increase the power budget $\bar{p}$ available to transmitters \cite{nuzman07}, so we can unlock fundamental bounds on the network performance by studying the asymptotic behavior of the solutions as $\bar{p}\to\infty$, as we have recently shown in the yet unpublished study in \cite{renatomaxmin}. However, to date, bounds of this type are rare in the literature, and the existing bounds on the utility are not asymptotically sharp for many well-known utility maximization problems.

Against this background, in this study we derive upper bounds for the utility and for the transmit energy efficiency (i.e., utility over power) achieved by solutions to utility maximization problems for a given power budget $\bar{p}$. Unlike the bounds in previous studies \cite{renatomaxmin}, the bounds derived here are asymptotically tight, and they are valid for a larger class of utility maximization problems. These bounds, which do not depend on any unknown constants, are particularly useful to determine power regions in which wireless networks are expected to be noise limited and interference limited. In addition, they reveal that the network utility and the energy efficiency scale as $\Theta(1)$ and $\Theta(1/\bar{p})$, respectively, as $\bar{p}\to\infty$ (in this study, big theta $\Theta$ is defined as in the standard family of Bachmann-Landau notations). The main tools for the analysis shown here are the results in \cite{nuzman07} and the concept of asymptotic functions~\cite{aus03}, which so far have received limited attention from the wireless community \cite{renato2016,renato2016maxmin,renatomaxmin}. To the best of our knowledge, we show for the first time properties of asymptotic functions associated with standard interference functions that are not necessarily convex or concave. We illustrate the theoretical findings by studying the utility obtained in a dense wireless network in a stadium, one of the use-cases considered for the fifth generation of wireless networks \cite{renato2016maxmin}.

\section{Mathematical Preliminaries}

 One of the main objectives of this section is to derive properties of asymptotic functions associated with standard interference functions. This section also clarifies much of the notation, and it reviews standard results and definitions that are required for the contributions in the next sections.

In more detail, by $\real_+$ and $\real_{++}$ we denote the set of non-negative reals and positive reals, respectively. The effective domain of a function $f:\real^N\to \real\cup\{\infty\}$ is given by $\mathrm{dom}f:=\{\signal{x}\in\real^N~|~f(\signal{x})<\infty\}$, and $f$ is proper if $\mathrm{dom}~f \neq\emptyset$. We say that $f$ is continuous when restricted to $C\subset \real^N$  if $(\forall\signal{x}\in C) (\forall(\signal{x}_n)_{n\in\Natural}\subset C) \lim_{n\to\infty} \signal{x}_n=\signal{x}\Rightarrow \lim_{n\to\infty} f(\signal{x}_n)=f(\signal{x})$, and we write $\lim_{n\to\infty}\signal{x}_n=\signal{x}$ if and only if $\lim_{n\to\infty}\|\signal{x}_n-\signal{x}\|=0$ for an arbitrary norm $\|\cdot\|$ (the choice of the norm is arbitrary because of the equivalence of norms in finite dimensional spaces). The notions of upper and lower semicontinuity for functions $f:\real^N\to \real \cup \{\infty\}$ restricted to sets $C\subset\real^N$ are defined similarly. Given $(\signal{x},\signal{y})\in\real^N\times\real^N$, vector inequalities such as $\signal{x}\le\signal{y}$ should be understood coordinate-wise. If $C\subset \real^N$ is a convex set, we say that a mapping $T:C\to\real^N:\signal{x}\mapsto[t_1(\signal{x}), \cdots, t_N(\signal{x})]$ is concave if, for every $i\in\{1,\ldots,N\}$, the function $t_i:C\to\real$ is concave. A norm $\|\cdot\|$ in $\real^N$ is said to be monotone if $(\forall\signal{x}\in\real^N)(\forall\signal{y}\in\real^N)$ $\signal{0}\le\signal{x}\le\signal{y}\Rightarrow\|\signal{x}\| \le \|\signal{y}\|$.

A fundamental mathematical tool used in this study is the analytic representation of asymptotic functions, which we state as a definition:

\begin{definition}
	\label{def.asymp_func} (\cite[Theorem~2.5.1]{aus03} Asymptotic function)
	The asymptotic function associated with a proper function $f:\real^N\to\real\cup \{\infty\}$ is the function given by
	\begin{align}
	\label{eq.asymp_func}
	\begin{array}{rcl}
	f_{\infty}:\real^N & \to & \real \cup \{\infty\} \\ 
	\signal{x} & \mapsto & \inf\left\{\displaystyle \liminf_{n\to\infty} \dfrac{f(h_n\signal{x}_n)}{h_n}~|~h_n\to\infty,~\signal{x}_n\to\signal{x} \right\},
	\end{array}
	\end{align} 
	where $(\signal{x}_n)_{n\in\Natural}$ and $(h_n)_{n\in\Natural}$ are sequences in $\real^N$ and $\real$, respectively. Equivalently, $f_{\infty}:\real^N  \to  \real \cup \{\infty\}:\signal{x}\mapsto \liminf_{h\to\infty,\signal{y}\to\signal{x}}{f(h\signal{y})}/{h}$.
\end{definition}

In particular, here we are mostly interested in computing asymptotic functions associated with standard interference functions,  defined as follows:

\begin{definition}
	\label{definition.inter_func}
	(\cite{yates95} Standard interference functions)  A function $f: \real^N \to \real_{++} \cup \{\infty\}$ is said to be a standard interference function if the following properties hold:\par 
	\begin{enumerate}
		\item ({\it Scalability}) $(\forall \signal{x}\in\real^N_+)$ $(\forall \alpha>1)$  $\alpha {f}(\signal{x})>f(\alpha\signal{x})$. \par
		\item ({\it Monotonicity}) $(\forall \signal{x}_1\in\real_+^N)$ $(\forall \signal{x}_2\in\real_+^N)$ $\signal{x}_1\ge\signal{x}_2 \Rightarrow{f}(\signal{x}_1)\ge f(\signal{x}_2)$. \par
		\item $f(\signal{x})=\infty$ if and only if $\signal{x}\notin\real_{+}^N$. 
	\end{enumerate}
	Given $N$ standard interference functions $t_i:\real^N \to \real_{++} \cup \{\infty\}$, $i=1,\ldots,N$, we call the mapping $T:\real^N_+\to\real_{++}^N:\signal{x}\mapsto[t_1(\signal{x}),\ldots, t_N(\signal{x})]$ a {\it standard interference mapping}.
\end{definition}

As shown in \cite[Proposition 1]{renato2016} and the references therein, functions that are positive and concave when restricted to the domain $\real_+^N$ are standard interference functions, but the converse does not hold in general. Therefore, as it will soon become clear, the bounds derived in the next section are more general than those in the recent unpublished work in \cite{renatomaxmin}. One of the key results that we further explore in the next section is the following:

\begin{fact}
	\label{fact.nuzman} 
	(\cite{nuzman07,krause01}) Let $\|\cdot\|$ be a monotone norm and $T:\real_{+}^N\to\real_{+}^N$ be either a standard interference mapping or a concave mapping that is positively homogeneous and, for every $\signal{x}\in\real_{+}^N\backslash\{\signal{0}\}$, there exists $m\in\Natural$ for which  $T^n(\signal{x})>\signal{0}$ for all $n\ge m$ (these mappings are said to be \emph{primitive}, and $T^n$ denotes the $n$-fold composition of $T$ with itself). Then each of the following holds:
	\begin{enumerate}
		\item[(i)] There exists a unique solution $(\signal{x}^\star, \lambda^\star)\in\real_{++}^N\times\real_{++}$ to the conditional eigenvalue problem
		\begin{problem}
			\label{problem.cond_eig}
			Find $(\signal{x}, \lambda)\in\real_{+}^N\times\real_{+}$ such that $T(\signal{x})=\lambda\signal{x}$ and $\|\signal{x}\|=1$.
		\end{problem}
		\item[(ii)] The sequence $(\signal{x}_n)_{n\in\Natural}\subset\real_{+}^N$ generated by 
		\begin{align}
		\label{eq.krause_iter}
		\signal{x}_{n+1} = T^\prime({\signal{x}_n}):=\dfrac{1}{\|T(\signal{x}_n)\|}T(\signal{x}_n),\quad\signal{x}_1\in\real_{++}^N,
		\end{align}
		converges geometrically to the uniquely existing vector $\signal{x}^\star\in\mathrm{Fix}(T^\prime):=\{\signal{x}\in\real_{+}^N~|~\signal{x}=T^\prime(\signal{x})\}$, which is also the vector $\signal{x}^\star$ of the tuple $(\signal{x}^\star,\lambda^\star)$ that solves Problem~\ref{problem.cond_eig}. Furthermore, the sequence $(\lambda_n:=\|T(\signal{x}_n)\|)_{n\in\Natural}\subset\real_{++}$ converges to $\lambda^\star$.
	\end{enumerate}
\end{fact}

We now proceed to the study of asymptotic functions associated with standard interference functions. We start with the following simple result.

\begin{lemma}
	\label{lemma.basic_properties}
	Let $f:\real^N \to \real_{++} \cup \{\infty\}$ be a standard interference function. Then we have:
	\begin{itemize}
		
		\item[(i)] $(\forall \signal{x}\in\real_{+}^N)(\forall \alpha\in~]0,1[)\quad f(\alpha\signal{x}) >\alpha f(\signal{x})$
		
		\item[(ii)] $(\forall \signal{x}\in\real_{+}^N)(\forall \alpha_1\in\real_{++})(\forall \alpha_2\in\real_{++}) \\
		\alpha_2 > \alpha_1 \Rightarrow \dfrac{1}{\alpha_1} f(\alpha_1\signal{x})	> \dfrac{1}{\alpha_2} f(\alpha_2\signal{x})>0$
		
		\item[(iii)] $(\forall \signal{x}\in\real_{+}^N)~ \lim_{h\to\infty} f(h\signal{x})/h\in\real_+$ (i.e., the limit always exists).
	\end{itemize}
\end{lemma}
\begin{proof}
	(i) The desired result follows from a simple implication of Definition~\ref{definition.inter_func}.1:
	\begin{align*}
	(\forall \signal{x}\in\real_{+}^N)(\forall \alpha\in~]0,1[) ~~ \dfrac{1}{\alpha} f(\alpha \signal{x}) > f\left(\dfrac{\alpha}{\alpha}\signal{x} \right)=f(\signal{x}).
	\end{align*}
	
	(ii) Let $\alpha_2 > \alpha_1 > 0$ and define $\beta:=\alpha_2/\alpha_1>1$. The proof is now a direct consequence of Definition~\ref{definition.inter_func}.1:
	\begin{align*}
	\dfrac{1}{ \alpha_1}f(\alpha_1\signal{x}) = \dfrac{\beta}{\beta \alpha_1}f(\alpha_1\signal{x}) > \dfrac{1}{\beta \alpha_1}f(\beta \alpha_1\signal{x})=\dfrac{1}{ \alpha_2}f(\alpha_2\signal{x})>0.
	\end{align*}
	
	(iii) By (ii), for every $\signal{x}\in\real_{+}^N$, the function $\real_{++}\to\real_{++}:h\mapsto f(h\signal{x})/{h}$ is monotonically decreasing and bounded below by zero, so the limit $\lim_{h\to\infty} f(h\signal{x})/h\in\real_+$ exists as claimed.	
\end{proof}

The next proposition shows important properties of asymptotic functions associated with standard interference functions. In particular, it establishes the connection between the limit in Lemma~\ref{lemma.basic_properties}(iii) and asymptotic functions. This connection is of practical significance because the limit in Lemma~\ref{lemma.basic_properties}(iii) is often easy to compute in the applications we consider here. In fact, Proposition~\ref{prop.af_properties}(i) shows that a useful simplification for the computation of $f_\infty$ when $f$ is convex \cite[Corollary~2.5.3]{aus03} is also available when $f$ is a standard interference function that is not necessarily convex or concave.

\begin{proposition}
	\label{prop.af_properties}
	 The asymptotic function $f_{\infty}:\real^N  \to  \real \cup \{\infty\}$ associated with a standard interference function $f:\real^N \to \real_{++} \cup \{\infty\}$ has the following properties:
	\begin{itemize}
		\item[(i)] $(\forall\signal{x}\in\real^N_+)~f_\infty(\signal{x}) = \lim_{h\to\infty} f(h \signal{x})/{h} \in\real_+$.
		
		\item[(ii)] $f_\infty$ is lower semicontinuous and positively homogeneous. If $f$ is in addition continuous when restricted to the non-negative orthant $\real_{+}^N$, then $f_\infty$ is continuous when restricted to the non-negative orthant $\real_{+}^N$.
		\item[(iii)] ({\it Monotonicity}) $(\forall \signal{x}_1\in\real_+^N)$ $(\forall \signal{x}_2\in\real_+^N)$ $\signal{x}_1\ge\signal{x}_2  \Rightarrow  {f}_\infty(\signal{x}_2)\le f_\infty(\signal{x}_1)$.
		\item[(iv)] Let  $\signal{x}\in\real^N_{+}$ be arbitrary. If $f$ is continuous when restricted to $\real_{+}^N$, then
		$f_\infty(\signal{x})=\lim_{n\to\infty}{f(h_n \signal{x}_n)}/{h_n}$ for all sequences $(\signal{x}_n)_{n\in\Natural}\subset\real_{+}^N$ and $(h_n)_{n\in\Natural}\subset\real_{++}$ such that $\lim_{n\to\infty}\signal{x}_n=\signal{x}$ and $\lim_{n\to\infty}h_n=\infty$.
		\item[(v)] If $f$ is also concave when restricted to $\real_{+}^N$, then $f_\infty$ is concave when restricted to $\real_{+}^N$. 
		
	\end{itemize}
	
\end{proposition}
\begin{proof}
	(i) The inequality $(\forall\signal{x}\in\real^N_+)~f_\infty(\signal{x}) \le \lim_{h\to\infty} f(h \signal{x})/{h}\in\real_+$ is immediate from Lemma~\ref{lemma.basic_properties}(iii) and the definition in \refeq{eq.asymp_func}, so it is sufficient to prove that $(\forall\signal{x}\in\real^N_+)~f_\infty(\signal{x}) \ge \lim_{h\to\infty} f(h \signal{x})/{h}$ to obtain the desired result.
	
	Let $\alpha\in~]0,1[$ be arbitrary. From the definition of the asymptotic function in \refeq{eq.asymp_func} and the definition of standard interference functions, we know that, for every $\signal{x}\in\real_{+}^N$, there exist a sequence $(\signal{x}_n)_{n\in\Natural}\subset \real_{+}^N$ and an  increasing sequence $(h_n)_{n\in\Natural}\subset \real_{++}$ such that $f_\infty(\signal{x})=\lim_{n\to\infty}{f(h_n\signal{x}_n)}/{h_n}$, $\lim_{n\to\infty} \signal{x}_n=\signal{x}$, and $\lim_{n\to\infty} h_n=\infty$. Therefore, as an implication of $\lim_{n\to\infty} \signal{x}_n=\signal{x}\in\real_+^N$, we have $\alpha\signal{x}\le\signal{x}_n$ for every $n\ge L$ with $L\in\Natural$ sufficiently large. Lemma~\ref{lemma.basic_properties}(i) and the monotonicity property of standard interference functions yield
	$
	(\forall n\ge L)
	0<\alpha{f(h_n\signal{x})}/{h_n}<{ f(\alpha h_n\signal{x})}/{h_n} \le { f(h_n\signal{x}_n)}/{h_n}$.	Taking the limit as $n\to\infty$ and considering Lemma~\ref{lemma.basic_properties}(iii), we verify that 
	\begin{align}
	\label{eq.limit_afunc}
	0\le \alpha \lim_{n\to\infty} \dfrac{f(h_n\signal{x})}{h_n} = \alpha {\lim_{h\to\infty}}\dfrac{f(h\signal{x})}{h} \le f_\infty(\signal{x}).
	\end{align}
	Since $\alpha$ can be made arbitrarily close to one, \refeq{eq.limit_afunc} proves that $f_\infty(\signal{x}) \ge  \lim_{h\to\infty} {f(h\signal{x})}/{h} \ge 0$, which completes the proof.

	(ii) First recall that asymptotic functions associated with proper functions are lower semicontinuous and positive homogeneous  \cite[Proposition~2.5.1(a)]{aus03}. Now, assume that $f$ is continuous when restricted to $\real_{+}^N$, and let $(h_n)_{n\in\Natural}\subset\real_{++}$ be an arbitrary sequence such that $\lim_{n\to\infty}h_n=\infty$. Define $g_n:\real_+^N\to\real_+:\signal{x}\mapsto {f(h_n \signal{x})}/{h_n}$, and note that the function $g_n$ is continuous for every $n\in\Natural$ because $f$ is continuous when restricted to $\real_+^N$ by assumption.  The property in (i) and Lemma~\ref{lemma.basic_properties}(ii)-(iii) imply that $(\forall \signal{x}\in\real_{+}^N) f_\infty(\signal{x})=\inf_{n\in\Natural} g_n(\signal{x})$. This shows that, restricted to $\real_{+}^N$, $f_\infty$ is the pointwise infimum of continuous functions, so  $f_\infty$ is upper semicontinuous, in addition to being lower semicontinuous as already shown. Therefore, $f_\infty$   restricted to $\real_{+}^N$ is continuous.

	(iii) Let $({h}_n)_{n\in\Natural}\subset\real_{++}$ be an arbitrary monotone sequence such that $\lim_{n\to\infty}h_n=\infty$. If $\signal{x_1}\ge\signal{x}_2\ge\signal{0}$, the monotonicity property of standard interference functions shows that 
	$
	(\forall n\in\Natural)~
	{f(h_n\signal{x}_2)}/{h_n} \le  {f(h_n\signal{x}_1)}/{h_n}.$ Now let $n\to\infty$ and use the property in (i) to obtain the desired result
$
f_\infty(\signal{x}_2) \le  f_\infty(\signal{x}_1).
$
	
	(iv) Let the arbitrary sequences $(\signal{x}_n)_{n\in\Natural}\subset\real_{+}^N$ and $(h_n)_{n\in\Natural}\subset\real_{++}$ satisfy $\lim_{n\to\infty}\signal{x}_n=\signal{x}$ and $\lim_{n\to\infty}h_n=\infty$. Denote by $\signal{1}\in\real^N$ the vector of ones. As a consequence of $\lim_{n\to\infty}\signal{x}_n=\signal{x}$, we know that 
	\begin{align}
	\label{eq.ineqxn}
	(\forall \epsilon>0)(\exists L\in\Natural)(\forall n\ge L)~ \signal{x}_n\le\signal{x}+\epsilon\signal{1}.
	\end{align}
	As a result,  for every $\epsilon>0$, we have
	\begin{multline}
	\label{eq.ineq_2ef}
	f_\infty(\signal{x}) \overset{(a)}{\le} \liminf_{n\in \Natural} \dfrac{1}{{h}_n} f({h}_n\signal{x}_n) \overset{(b)}{\le} \limsup_{n\in \Natural} \dfrac{1}{{h}_n} f({h}_n\signal{x}_n) \\ \overset{(c)}{\le} \limsup_{n\in \Natural} \dfrac{1}{{h}_n} f({h}_n(\signal{x}+\epsilon\signal{1})) \overset{(d)}{=} f_\infty(\signal{x}+\epsilon\signal{1}),
	\end{multline}
	where (a) follows from \refeq{eq.asymp_func}, (b) is a basic property of limits, (c) follows from \refeq{eq.ineqxn} and  monotonicity of $f$, and (d) is a consequence of the property we proved in (i).  By assumption, $f$ restricted to $\real_{+}^N$ is continuous, so $f_\infty$ restricted to $\real_{+}^N$ is also continuous  as shown in (ii). Therefore, by \refeq{eq.ineq_2ef} and continuity of $f_\infty$ when restricted to $\real_+^N$, we have
	\begin{multline*}
	f_\infty(\signal{x}) {\le} \liminf_{n\in \Natural} \dfrac{1}{{h}_n} f({h}_n\signal{x}_n) {\le} \limsup_{n\in \Natural} \dfrac{1}{{h}_n} f({h}_n\signal{x}_n) \\ {\le} \lim_{\epsilon\to 0^+} f_\infty(\signal{x}+\epsilon\signal{1})=f_\infty(\signal{x}),
	\end{multline*}	
	which implies that $\lim_{n\in \Natural} f({h}_n\signal{x}_n)/{{h}_n}=f_\infty(\signal{x})$.
	
	(v) Let $g_n$ be the function defined in the proof of (ii). Note that, if $f$ is concave in $\real_{+}^N$, then $g_n$ is also concave in $\real_+^N$ for every $n\in\Natural$. In (ii) we have  showed that $(\forall\signal{x}\in\real_+^N)~f_\infty(\signal{x})=\inf_{n\in\Natural}g_n(\signal{x})$, so, in $\real_{+}^N$, $f_\infty$ is the pointwise infimum of a family of concave functions, and the proof is complete because concavity is preserved by taking the pointwise infimum.  
\end{proof}

We end this section by defining a mapping that plays a crucial role in the analysis of network utility maximization problems.

\begin{definition}
	\label{def.amap} (Asymptotic mappings)
	Let the function $t^{(i)}:\real^N\to\real_{++}\cup\{\infty\}$ be a standard interference function for each $i\in\{0,\ldots,N\}$. Given a standard interference mapping $T: \real^N_+ \to\real^N_{++}:\signal{x}\mapsto [t^{(1)}(\signal{x}),\cdots,t^{(N)}(\signal{x})]$, the asymptotic mapping associated with $T$ is given by 
$
	T_\infty:\real^N_+\to\real^N_+:\signal{x}\mapsto [{t}^{(1)}_\infty(\signal{x}),\cdots,{t}^{(N)}_\infty(\signal{x})],
$
	where, for each $i\in\{1,\cdots,N\}$, ${t}^{(i)}_\infty$ is the asymptotic function associated with $t^{(i)}$.
\end{definition}

\section{Properties of Utility Maximization Problems}

As shown in \cite{nuzman07,renato2016maxmin,sun2014,renatomaxmin} and the references therein, a large class of (weighted max-min) utility maximization problems in wireless networks are particular instances of the following canonical optimization problem:

\begin{problem}
	\label{problem.canonical}
	(Canonical form of the network utility maximization problem)
	\begin{align}
	\label{eq.canonical}
	\begin{array}{lll}
	\mathrm{maximize}_{\signal{p}, c} & c \\
	\mathrm{subject~to} & \signal{p}\in \mathrm{Fix}(cT):=\left\{\signal{p}\in\real_{+}^N~|~\signal{p}=cT(\signal{p})\right\} \\
	& \|\signal{p}\|_a \le \bar{p} \\
	& \signal{p}\in\real_{+}^N, c\in\real_{++},
	\end{array}
	\end{align}
	where $\bar{p}\in \real_{++}$ is a design parameter hereafter called power budget,  $\|\cdot\|_a$ is an arbitrary monotone norm, and $T:\real_{+}^N\to\real_{++}^N$ is an arbitrary standard interference mapping. 
\end{problem}

The main objective of this section is to study selected properties of the solution to Problem~\ref{problem.canonical} as a function of the power budget $\bar{p}$. In particular,  these properties enable us to characterize the power regions of a noise limited regime and an interference limited regime in wireless networks. The results in the following not only generalize those of \cite{renatomaxmin} to a larger class of mappings $T$, but they also sharpen the bounds of that study. They are obtained from a deep connection between Problem~\ref{problem.canonical} and conditional eigenvalues of nonlinear mappings:

\begin{fact}
	\label{fact.nuzman_main}
	\cite{nuzman07} Consider the assumptions and definitions in Problem~\ref{problem.canonical}, and denote by $(\signal{p}_{\bar{p}}, c_{\bar{p}})\in\real_{++}^N\times\real_{++}$ a solution to this problem for a given power budget $\bar{p}$. Then:
	\begin{itemize}
		\item[(i)] The solution $(\signal{p}_{\bar{p}}, c_{\bar{p}})\in\real_{++}^N\times\real_{++}$ always exists, and it is unique.
		\item[(ii)] Let $(\signal{x}_{\bar{p}}, \lambda_{\bar{p}})\in\real_{++}^N\times\real_{++}$ be the solution to the following conditional eigenvalue problem: 
		\begin{problem}
			\label{problem.sol_prob3}
			Find $(\signal{x}_{\bar{p}}, \lambda_{\bar{p}})\in\real_{++}^N\times\real_{++}$ such that $T(\signal{x}_{\bar{p}}) = \lambda_{\bar{p}} \signal{x}_{\bar{p}}$ and $\|\signal{x}_{\bar{p}}\|=1$, where $\|\cdot\|$ denotes the monotone norm $\|\signal{x}\|:=\|\signal{x}\|_a/\bar{p}$.
		\end{problem}		
		Then $\signal{p}_{\bar{p}} = \signal{x}_{\bar{p}}$ and $c_{\bar{p}} = 1/\lambda_{\bar{p}}$.

	\item[(iii)] The function $\real_{++}\to\real_{++}:\bar{p} \mapsto c_{\bar{p}}$ is monotonically increasing; i.e., $\bar{p}_1>\bar{p}_2>0$ implies $c_{\bar{p}_1}>c_{\bar{p}_2}>0$.
	\item[(iv)] The function $\real_{++}\to\real_{++}^N:\bar{p} \mapsto \signal{p}_{\bar{p}}$ is monotonically increasing in each coordinate; i.e., $\bar{p}_1>\bar{p}_2>0$ implies $\signal{p}_{\bar{p}_1}>\signal{p}_{\bar{p}_2}>\signal{0}$.
		\end{itemize}
\end{fact}

An important implication of Fact~\ref{fact.nuzman_main}(ii) is that the simple iterative scheme in \refeq{eq.krause_iter} is able to solve Problem~\ref{problem.canonical}. Fact~\ref{fact.nuzman_main} also motivates the use of the following functions \cite{renatomaxmin}:

\begin{definition}
	\label{def.mm_ee}
	(Utility, power, and $\|\cdot\|_b$-energy efficiency functions) Denote by  $(\signal{p}_{\bar{p}},~c_{\bar{p}})\in\real_{++}^N\times\real_{++}$ the solution to Problem~\ref{problem.canonical} for a given power budget $\bar{p}\in\real_{++}$. The utility and power functions are defined by, respectively, $U:\real_{++}\to\real_{++}:\bar{p}\mapsto c_{\bar{p}}$ and $P:\real_{++}\to\real_{++}^N:\bar{p}\mapsto \signal{p}_{\bar{p}}$. In turn, given a monotone norm $\|\cdot\|_b$, the $\|\cdot\|_b$-energy efficiency function is defined by $E:\real_{++}\to\real_{++}:\bar{p}\mapsto U(\bar{p})/\|P(\bar{p})\|_b$, and note that $(\forall \bar{p}\in\real_{++})~ E(\bar{p})=1/\|T(\signal{p}_{\bar{p}})\|_b=c_{\bar{p}}/\|\signal{p}_{\bar{p}}\|_b$.
\end{definition}

 By Fact~\ref{fact.nuzman_main}(iii)-(iv),  the utility function $U$ and each coordinate of the power function $P$ are monotonically increasing. However, in the next lemma, we show that the utility cannot grow faster than the transmit power.
 
\begin{lemma} 
	\label{lemma.e_monotone} The $\|\cdot\|_b$-efficiency function $E:\real_{++}\to\real_{++}$ is non-increasing; i.e., ${\bar{p}_1}>{\bar{p}_2}>0$ implies $E(\bar{p}_1) \le E(\bar{p}_2)$.
\end{lemma}
\begin{proof}
The result follows from  
\begin{multline*}
 {\bar{p}_1}>{\bar{p}_2}>0 \overset{(a)}{\Rightarrow} P(\bar{p}_1) > P(\bar{p}_2)  \overset{(b)}{\Rightarrow} T(P(\bar{p}_1)) \ge T(P(\bar{p}_2)) \\
 \overset{(c)}{\Rightarrow} \|T(P(\bar{p}_1))\|_b \ge \|T(P(\bar{p}_2))\|_b >0 \overset{(d)}{\Leftrightarrow} E(\bar{p}_1) \le E(\bar{p}_2),
\end{multline*}
where (a) is an implication of Fact~\ref{fact.nuzman_main}(iv), (b) is a consequence of the monotonicity of standard interference functions, (c) results from the monotonicity of the norm $\|\cdot\|_b$ and positivity of $T$, and (d) is obtained from a basic property of $E$ shown in Definition~\ref{def.mm_ee}.
\end{proof}

We can also prove that the functions $U$, $P$, and $E$ in Definition~\ref{def.mm_ee} are continuous in $\real_{++}$, but we omit the proof owing to the space limitation (it is similar to that shown in \cite{renatomaxmin} for utility optimization problems with concave mappings $T$). We are now ready to show that spectral properties of asymptotic mappings are useful to characterize the performance of wireless networks in a unified way. To this end, we need the following technical result:

\begin{proposition}
	\label{prop.acondv}
	Consider the assumptions and definitions of Problem~\ref{problem.canonical}. For notational convenience, let $(\signal{p}_{\bar{p}}, \lambda_{\bar{p}}):=(P(\bar{p}), 1/U(\bar{p}))\in\real_{++}^N\times\real_{+}$ be the solution to Problem~\ref{problem.sol_prob3} for a given power budget $\bar{p}\in\real_{++}$ (see Fact~\ref{fact.nuzman}), and denote by $T_\infty:\real_{+}^N\to\real_{+}^N$ the asymptotic mapping associated with the standard interference mapping $T:\real_{+}^N\to\real_{++}^N$. Then:
	\begin{itemize}
		\item[(i)] The limit $\lim_{\bar{p}\to\infty}\lambda_{\bar{p}}=:\lambda_\infty\ge 0$ exists.
		\item[(ii)] Let the scalar $\lambda_\infty$ be as defined in (i), and assume that $T$ is continuous. In addition, let $(\bar{p}_n)_{n\in\Natural}\subset \real_{++}$ denote an arbitrary monotonically increasing sequence satisfying $\lim_{n\to\infty} \bar{p}_n = \infty$, and define $\signal{x}_{n} := (1/\|\signal{p}_{\bar{p}_n}\|_a)\signal{p}_{\bar{p}_n}$. If $\signal{x}_{\infty}\in\real_+^N$ is an accumulation point of the bounded sequence $(\signal{x}_n)_{n\in\Natural}\subset\real_{++}^N$, then the tuple $(\signal{x}_\infty,\lambda_\infty)$ solves the following conditional eigenvalue problem:
	
		\begin{problem}
			\label{problem.eival_tinf}
			Find $(\signal{x}, \lambda)\in\real_{+}^N\times\real_{+}$ such that $T_\infty(\signal{x}) = \lambda \signal{x}$ and $\|\signal{x}\|_a = 1$.
		\end{problem}	
		
		\item[(iii)] Consider the assumptions and the notation in (i) and (ii). If Problem~\ref{problem.eival_tinf} has a unique solution denoted by $(\signal{x}^\prime, \lambda^\prime)\in\real_{++}^N\times\real_{++}$, then $\lim_{n\to\infty}\signal{x}_n=\signal{x}_\infty=\signal{x}^\prime$ and $\lambda^\prime=\lambda_\infty$.

	\end{itemize}
\end{proposition}
\begin{proof}
	 (i) The limit $\lim_{n\to\infty} \lambda_{\bar{p}_n} =: \lambda_\infty\ge 0$ exists because the function $\real_{++}\to\real_{++}:\bar{p} \mapsto \lambda_{\bar{p}}$ is monotonically decreasing (and bounded below by zero) by Fact~\ref{fact.nuzman_main}(iii).

	 (ii) First note that $T_\infty$ is continuous as a consequence of Proposition~\ref{prop.af_properties}(ii). Since $(\forall n\in \Natural)~T(\signal{p}_{\bar{p}_n}) = \lambda_{\bar{p}_n} \signal{p}_{\bar{p}_n}$ and $\|\signal{p}_{\bar{p}_n}\|_a = \bar{p}_n>0$, we have
	 \begin{multline}
	 \label{eq.xn}
	 (\forall n\in \Natural) \lambda_{\bar{p}_n} \signal{x}_n = \dfrac{\lambda_{\bar{p}_n}}{\|\signal{p}_{\bar{p}_n}\|_a} \signal{p}_{\bar{p}_n} 
	 = \dfrac{1}{\|\signal{p}_{\bar{p}_n}\|_a} T\left(\dfrac{\|\signal{p}_{\bar{p}_n}\|_a}{\|\signal{p}_{\bar{p}_n}\|_a} \signal{p}_{\bar{p}_n}\right) \\ = \dfrac{1}{\bar{p}_n} T(\bar{p}_n\signal{x}_n).
	 \end{multline}

 Now, if $\signal{x}_\infty$ is an accumulation point of the bounded sequence $(\signal{x}_n)_{n\in\Natural}\subset\real_{++}^N$, then there exists a convergent subsequence $(\signal{x}_n)_{n\in K}$, $K\subset\Natural$, such that $\lim_{n\in K} \signal{x}_{n} = \signal{x}_\infty\in\real_+^N$. Therefore, from the result in (i) and \refeq{eq.xn}, we have
	 \begin{align*}
	 	 \lambda_\infty \signal{x}_\infty = \lim_{n\in K} \dfrac{1}{\bar{p}_n} T(\bar{p}_n\signal{x}_n) = T_\infty(\signal{x}_\infty),
	 \end{align*}
	 where the last equality follows from Proposition~\ref{prop.af_properties}(iv) and continuity of $T_\infty$. 
We now conclude the proof by noticing that  $\|\signal{x}_\infty\|_a = \lim_{n\in K} \|\signal{x}_n\|_a=1$.

(iii) If $(\signal{x}^\prime, \lambda^\prime)\in\real_{++}^N\times\real_{++}$ is the unique solution to Problem~\ref{problem.eival_tinf}, then, as an immediate consequence of the result in (ii), the only accumulation point of the bounded sequence $(\signal{x}_n)_{n\in\Natural}$ is $\signal{x}^\prime$, which implies that $\lim_{n\to\infty}\signal{x}_n=\signal{x}_\infty=\signal{x}^\prime$. As a result, $T_\infty(\signal{x}_\infty) =\lambda^\prime \signal{x}_\infty\in\real_{++}^N$, which also proves that $\lambda^\prime=\lambda_\infty$ by considering the result in (ii).\end{proof}

The assumption of uniqueness and positivity of the solution to Problem~\ref{problem.eival_tinf} in Proposition~\ref{prop.acondv}(iii) is valid in many utility maximization problems. In particular, Fact~\ref{fact.nuzman} shows sufficient conditions that are easily verifiable, and it also shows a simple fixed point algorithm able to solve Problem~\ref{problem.eival_tinf}. For the large class of utility maximization problems for which this assumption is valid, we can obtain simple performance bounds that are tight and fast to compute:

\begin{proposition}
\label{prop.bounds}
Let the assumptions in Proposition~\ref{prop.acondv}(iii) be valid, and consider the functions in Definition~\ref{def.mm_ee}. Furthermore, denote by $(\signal{x}_\infty, \lambda_\infty)\in\real^N_{++}\times\real_{++}$ the solution to Problem~\ref{problem.eival_tinf}. Then the following holds:
\begin{itemize}
		\item[(i)] $\sup_{\bar{p}>0} U(\bar{p}) = \lim_{\bar{p}\to \infty} U(\bar{p}) = 1/\lambda_\infty$ and $\sup_{\bar{p}>0} E(\bar{p}) = \lim_{\bar{p}\to 0^+} E(\bar{p}) = 1/\|T(\signal{0})\|_b$.
	\item[(ii)] $(\forall \bar{p}\in\real_{++})$ 
	\begin{align*} U(\bar{p}) 
	\le  \begin{cases} \bar{p}/\|T(\signal{0})\|_a,&\text{if } \bar{p}\le\|T(\signal{0})\|_a/\lambda_\infty\\ 1/\lambda_\infty,&\text{otherwise.}\end{cases}
	\end{align*}	
	\item[(iii)] $(\forall \bar{p}\in\real_{++}) E(\bar{p}) \le \min\{1/\|T(\signal{0})\|_b, \alpha/(\lambda_\infty~\bar{p})\}$, where $\alpha\in\real_{++}$ is any scalar satisfying $(\forall\signal{x}\in\real^N)~\|\signal{x}\|_a \le \alpha \|\signal{x}\|_b$ (such a scalar always exists because of the equivalence of norms in finite dimensional spaces).
	
	\item[(iv)] $U(\bar{p})\in\Theta(1)$ and $E(\bar{p})\in\Theta(1/\bar{p})$ as $\bar{p}\to\infty$, where big theta $\Theta$ is defined as in the standard family of Bachmann-Landau notations.
\end{itemize}
\end{proposition}
\begin{proof}
	(i) The function $U$ is monotonically increasing by Fact~\ref{fact.nuzman_main}(iii) and $\lim_{\bar{p}\to\infty} U(\bar{p})= 1/\lambda_\infty$ by Proposition~\ref{prop.acondv}(iii),  so $\sup_{\bar{p}>0} U(\bar{p}) = \lim_{\bar{p}\to \infty} U(\bar{p}) = 1/\lambda_\infty$. 

Now, let $(\bar{p}_n)_{n\in\Natural}\subset \real_{++}$ be an arbitrary sequence such that $\lim_{n\to\infty}\bar{p}_n = 0$. To prove that $\lim_{\bar{p}\to 0^+} E(\bar{p}) = \lim_{\bar{p}\to 0^+} {1}/{\|T(P(\bar{p})\|_b} = {1}/{\|T(\signal{0})\|_b}$, we only have to show that $\lim_{n\to\infty} {1}/{\|T(P(\bar{p}_n))\|_b} = {1}/{\|T(\signal{0})\|_b}$. By Fact~\ref{fact.nuzman_main}(ii), we have $\lim_{n\to\infty} \|P(\bar{p}_n)\|_a = \lim_{n\to\infty} \bar{p}_n = {0}$, and thus $\lim_{n\to\infty} P(\bar{p}_n)=\signal{0}$. Therefore, by continuity and positivity of $T$, we obtain $\lim_{n\to\infty} E(\bar{p}_n) = \lim_{n\to\infty} {1}/{\|T(P(\bar{p}_n))\|_b}= {1}/{\|T(\signal{0})\|_b}<\infty$. 	The equality $\sup_{\bar{p}>0} E(\bar{p}) = \lim_{\bar{p}\to 0^+} E(\bar{p}) = {1}/{\|T(\signal{0})\|_b}$ is now immediate from Lemma~\ref{lemma.e_monotone}, and the proof of (i) is complete.

(ii) By Fact~\ref{fact.nuzman_main}(ii), positivity of $U$ and $T$, and monotonicity of $T$, $\|\cdot\|_a$, and $P$, we have 
\begin{align*}
(\forall\bar{p}>0)~ 0< U(\bar{p}) \|T(\signal{0})\|_a \le  U(\bar{p}) \|T(P(\bar{p}))\|_a   = \bar{p},
\end{align*}
 and thus $(\forall\bar{p}>0)~U(\bar{p})\le \bar{p}/ \|T(\signal{0})\|_a$. Furthermore, by (i), we also have $(\forall\bar{p}>0)~U(\bar{p})\le 1/\lambda_\infty$. Combining these two last inequalities, we obtain the desired result $(\forall \bar{p}\in\real_{++}) U(\bar{p}) \le  \min\{\bar{p}/\|T(\signal{0})\|_a, 1/\lambda_\infty\}$.
 
 (iii)  By (i), Fact~\ref{fact.nuzman_main}(ii), and the definition of the energy efficiency function, we deduce for every $\bar{p}>0$:
 \begin{align*}
  E(\bar{p})=\dfrac{U(\bar{p})}{\|P(\bar{p}))\|_b}\le \dfrac{\alpha  U(\bar{p})}{\|P(\bar{p}))\|_a}  = \dfrac{\alpha U(\bar{p})}{\bar{p}} \le \dfrac{\alpha}{\lambda_\infty \bar{p}}. 
 \end{align*}
The desired result is now obtained by combining the previous bound with the bound $(\forall \bar{p}>0)~E(\bar{p}) \le 1/\|T(\signal{0})\|_b$, which is immediate from (i).

(iv) The relation  $U(\bar{p})\in\Theta(1)$ is immediate from $\lim_{\bar{p}\to \infty} U(\bar{p}) = 1/\lambda_\infty$, as shown in (i). To prove that $E(\bar{p})\in\Theta(1/\bar{p})$, recall that, from the equivalence of norms in finite dimensional spaces, there exists a scalar $\beta\in\real_{++}$ such that $(\forall \bar{p}>0)~\|P(\bar{p})\|_\infty \le \beta \|P(\bar{p})\|_a=\beta\bar{p}$, which  implies $(\forall \bar{p}>0)~P(\bar{p}) \le \beta\bar{p}\signal{1}$. Now we use the bound in (iii), the monotonicity of the norm $\|\cdot\|_a$, and the monotonicity and scalability properties of standard interference functions to verify that  $(\forall \bar{p}>1)~{\alpha}/{(\lambda_\infty \bar{p})} \ge {E}(\bar{p}) = {1}/{\|T(P(\bar{p}))\|_b} \ge {1}/{
	\|T(\beta \bar{p} \signal{1})\|_b} \ge {1}/{(\bar{p} \|T(\beta  \signal{1})\|_b)}$, which implies $E(\bar{p})\in\Theta(1/\bar{p})$ as $\bar{p}\to\infty$.
\end{proof}

 Proposition~\ref{prop.bounds}(ii) motivates the definition of the following operating point, which is an improvement over that in \cite{renatomaxmin} because the utility bounds in Proposition~\ref{prop.bounds} are tight, and we consider a larger class of utility maximization problems:
 
 \begin{definition}
 	If the assumptions of Proposition~\ref{prop.bounds} are valid, we say that the network operates in the \emph{low power regime} if $\bar{p} \le \bar{p}_\mathrm{T}$ or in the \emph{high power regime} if $\bar{p} > \bar{p}_\mathrm{T}$, where the power budget $\bar{p}_\mathrm{T}:=\|T(\signal{0})\|_a/\lambda_\infty$ is the \emph{transition point}.  
 \end{definition}

In practice, the transition point is the power budget in which networks are typically transitioning from a regime where the performance is \emph{limited by noise} to a regime where the performance is \emph{limited by interference.}

\section{Numerical example and final remarks}
	
	We apply the above results to the utility maximization problem described in \cite{renato2016maxmin}. Briefly, the objective is to maximize the minimum  downlink rate in an OFDMA-based network, and the optimization variables are the rates, the transmit power, and the load (i.e., the fraction of resource blocks used for transmission) at the base stations. As shown in \cite{renato2016maxmin}, this problem can be written in the canonical form in Problem~\ref{problem.canonical}. In the simulation, we use the same dense network used to produce \cite[Fig.~2]{renato2016maxmin}, with the only difference that here the noise power spectral density is fixed to -154~dBm/Hz, and we vary the power budget. For brevity, we refer the readers to \cite[Sect.~V.B]{renato2016maxmin} for further details. Fig.~\ref{fig.utility} shows the utility obtained in the simulations together with the utility bound in Proposition~\ref{prop.bounds}(ii) (for completeness, we show the energy efficiency in Fig.~\ref{fig.ee}). All conditional eigenvalue problems have been solved with the fixed point iteration in \refeq{eq.krause_iter}. In Fig.~\ref{fig.utility}, we clearly observe that, for a power budget above the transition point, increasing the budget by orders of magnitude brings only marginal gains in utility, an indication that the performance is limited by interference. In contrast, for a power budget below the transition point, the gain in utility is close to linear, which is an indication that the performance is limited by noise. We also note that the simple bound in  Proposition~\ref{prop.bounds}(ii), which requires the solution of only one conditional eigenvalue problem, provides us with a good estimate of the actual network performance for any value of $\bar{p}$. This observation is not necessarily surprising because the bound for $U$ is asymptotically sharp as $\bar{p}\to 0^+$ and as $\bar{p}\to \infty$. Similar results have been obtained for different network utility maximization problems (e.g., \cite{sun2014}\cite[Sect.~IV]{renatomaxmin}), and they will be reported elsewhere.
	  
	\begin{figure}
		\begin{center}
			\includegraphics[width=\columnwidth]{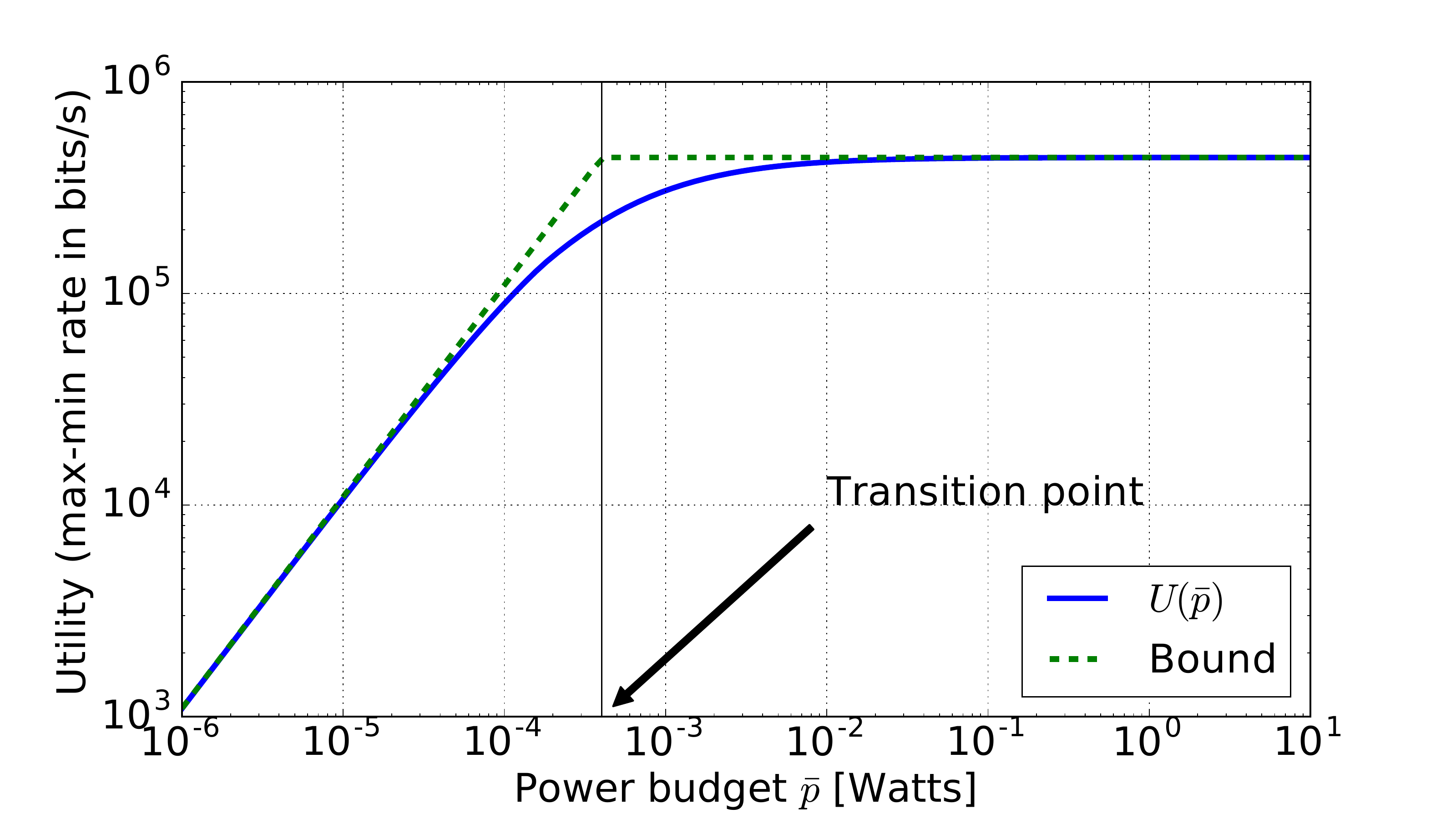}
			\caption{Network utility as a function of the power budget $\bar{p}$ for the problem described in \cite[Sect.~V-B]{renatomaxmin}.}
			\label{fig.utility}
		\end{center}

	\end{figure}
	
	\begin{figure}
		\begin{center}
			\includegraphics[width=\columnwidth]{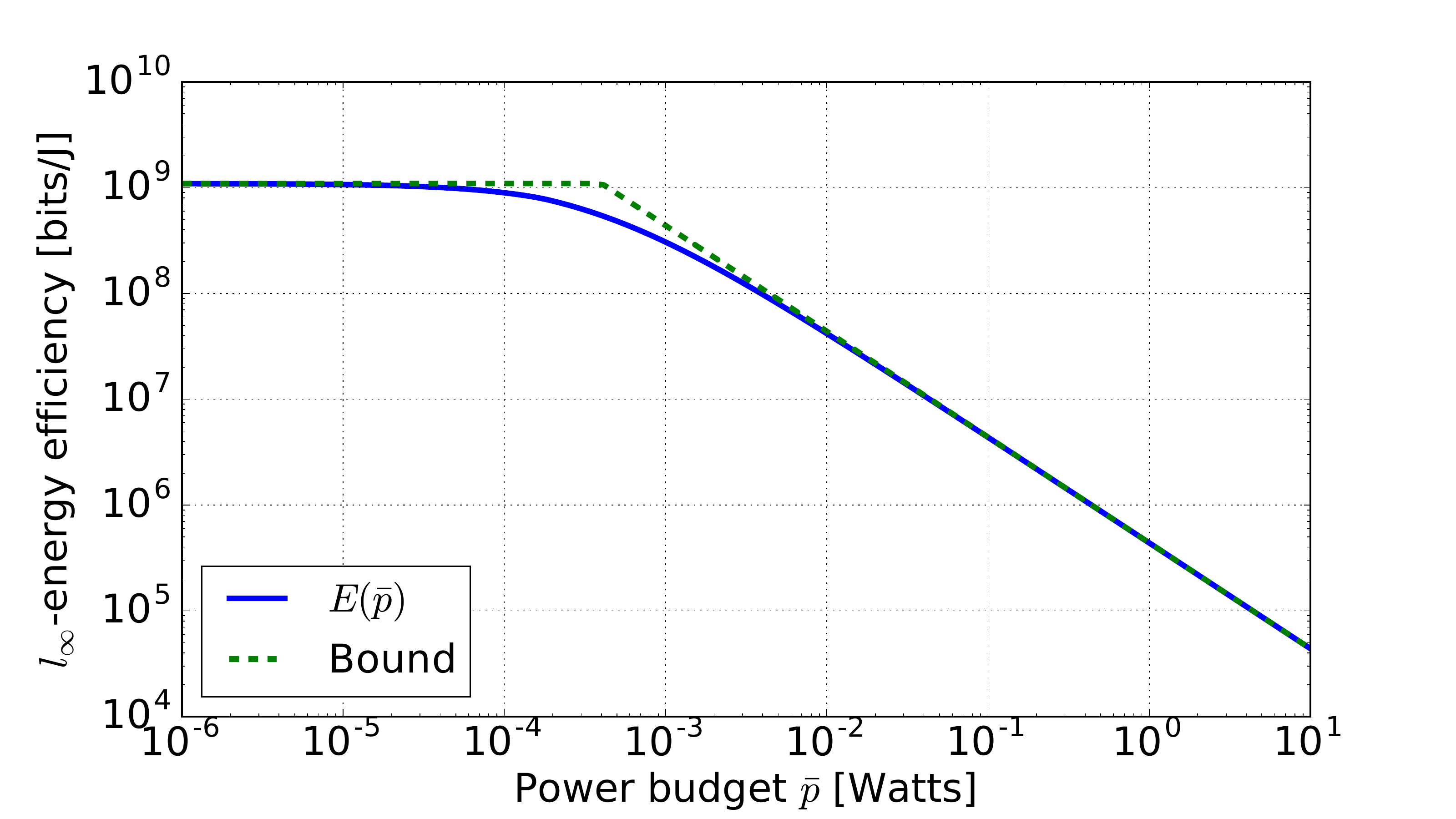}
			\caption{Energy efficiency as a function of the power budget $\bar{p}$ for the problem described in \cite[Sect.~V-B]{renatomaxmin}.}
			\label{fig.ee}
		\end{center}

	\end{figure}
	
\bibliographystyle{IEEEtran}
\bibliography{IEEEabrv,references}

\end{document}